\documentclass[12pt]{article}

\usepackage{amsmath,amsfonts,amssymb,amsthm}
\usepackage{setspace,graphicx}
\usepackage{natbib}
\usepackage{bm}
\usepackage{ifthen}
\usepackage{verbatim}
\usepackage{color}
\usepackage{dsfont}
\usepackage{tikz}
\usepackage{setspace,graphicx}
\usepackage{sgame}
\usepackage{bm}
\usepackage{multirow}
\usepackage{ifthen}
\usepackage{verbatim}
\usepackage{xcolor}
\usetikzlibrary{arrows.meta}
\usepackage{mathabx}
\usepackage{relsize}
\usepackage{pgfplots}
\pgfplotsset{compat=1.18}

\usepackage[margin=1.25in]{geometry}

\theoremstyle{plain}
\theoremstyle{definition}

  \newtheorem{corollary}{Corollary}
  \newtheorem{definition}{Definition}
  
  \newtheorem{example}{Example}
  \newtheorem{lemma}{Lemma}
  
  \newtheorem{proposition}{Proposition}

  \theoremstyle{remark}

\def\EE{\mathbb{E}}
\def\RR{\mathbb{R}}
\DeclareMathOperator*{\argmax}{arg\,max}

\begin{document}

\title{Sequential Non-Bayesian Persuasion}

\author{Yaron Azrieli and Rachana Das\thanks{Department of Economics, The Ohio State University, 1945 North High street, Columbus, OH 43210. azrieli.2@osu.edu and das.402@osu.edu. We are especially grateful for the comments of two anonymous referees and an Associate Editor of Games and Economic Behavior. We also thank PJ Healy, Masanori Kobayashi, Jim Peck, and Fedor Sandomirskiy, as well as audiences at several conferences where this work was presented, for their feedback.}}

\maketitle



\bigskip

\begin{abstract}
We study a model of persuasion in which the receiver's updating rule systematically distorts Bayesian posteriors. While in the classic Bayesian case providing information sequentially is never valuable, we show that the sender gains from sequential persuasion in several of the environments considered in the literature on strategic information transmission. Our proofs are constructive and reveal how properties of the bias interact with the payoff environment to make dynamic persuasion beneficial.  

\vspace{4mm}
\noindent Keywords: Bayesian persuasion, biased belief updating, sequential persuasion.


\end{abstract}

\newpage

\section{Introduction}\label{sec-introduction}

Following the work of \cite{KamenicaGentzkow2011}, Bayesian Persuasion (BP) has become a focus of economic theory research. As its name suggests, a pillar of this literature is Bayes rule: Upon observing new information, agents `rationally' update beliefs as Bayes rule prescribes. While this is a normatively appealing assumption, there is ample evidence that people regularly deviate from this rule and are influenced by various types of biases \citep{benjamin2019errors, ortoleva2024alternatives}. In a recent paper, \cite{de2022non} study BP while relaxing the assumption that the receiver follows Bayes rule; instead, they assume that the receiver is characterized by a systematic bias that determines their posterior as a function of the prior and the correct Bayesian posterior; among other things, they show that the concavification method used to solve the standard BP problem still applies in this setting.

In this paper, we revisit the framework of \cite{de2022non} but expand the set of persuasion strategies to include sequential information provision. It has been shown experimentally that releasing information over time and repeating information influence decisions in various contexts \citep{carlson2006leader, schulz2016preference}, but such persuasion strategies are redundant if decision makers are Bayesian: Any sequential strategy can be replaced by a single information structure that induces the same distribution over posteriors and hence the same expected payoffs for the sender and the receiver \citep{kamenica2009bayesian}. As we argue below, a departure from Bayesianism can explain why sequential persuasion is more effective than a one-shot release of information. 

Our main results compare the maximal attainable payoff for the sender with a single period of persuasion to their payoff with an unlimited (but finite) number of signals. A first simple observation is that with multiple signals we can no longer describe the bias of the receiver as a function of the correct Bayesian posterior: There are different combinations of signals that lead to the same Bayesian posterior (held by the sender) but to different biased posteriors for the receiver. As a result, the concavification approach is not applicable once sequential persuasion is allowed. Consequently, it is typically challenging to calculate the optimal payoff for the sender with sequential persuasion even in simple environments where concavification can be used to characterize the sender's payoff with one-shot persuasion as in \citet{de2022non}.\footnote{See the discussion in Subsection \ref{subsec-first-best} for more on this point.} Nevertheless, we show that in many of the classic environments considered in the literature sequential persuasion is valuable. We say that \emph{the sender gains from sequential persuasion} if there exists a finite sequence of experiments\footnote{More precisely, in a sequential persuasion strategy the sender can condition the experiment chosen in the $n$th period on the history of signal realizations in all previous $n-1$ periods. The sender can also decide whether to stop or provide additional information as a function of the history.} that yields a strictly higher expected payoff to the sender than the maximal expected payoff attainable with one experiment. 

The key property that distinguishes our receiver from a Bayesian one is that their belief remains interior regardless of the information they obtain. In other words, no state is ever completely ruled out. 
Another property of the updating rule that we use for our results is eventual learning: If the receiver is repeatedly presented with information that induces a given Bayesian posterior, then their belief converges to that posterior.\footnote{For some of the results we only need this to be true for beliefs that assign probability 1 to a certain state. See Definition \ref{def-eventual-learning} below.} Thus, while the receiver is not Bayesian immediately, their belief is similar to that of a Bayesian after a prolonged, constant persuasion. To illustrate the updating process and the underlying mechanisms driving our results, our examples focus on a parametric family of biases known as ``conservative Bayesianism'' in which the receiver's posterior given prior $\mu_0$ and Bayesian posterior $\mu$ is given by      
$$\alpha \mu_0 + (1 - \alpha)\mu,$$ 
where $\alpha\in (0,1)$ captures the strength of the bias.\footnote{This updating rule is interior (assuming the prior is interior) and satisfies eventual learning.} Under-updating of beliefs in light of new information is well-documented in the experimental literature \citep[Section 8]{benjamin2019errors} and can explain anomalies in the field \citep[Section 3]{ortoleva2024alternatives}. Furthermore, this specific form of bias is relatively easy to work with, has axiomatic foundations \citep{epstein2006axiomatic, kovach2021conservative}, and it preserves the martingale property of Bayesian updating that is known to be crucial in BP models. 

The first result, Proposition \ref{prop-transparent-2-actions}, concerns environments where the sender has state-independent preferences and the receiver has only two actions, such as the `Prosecutor-Judge' example of \citet{KamenicaGentzkow2011}. Assuming that the bias is interior and satisfies eventual learning, we prove that the sender gains from sequential persuasion in all such environments if and only if the updating rule is `weakly Bayes plausible', which means that for every information structure the prior is in the convex hull of the biased posteriors. Notice that conservative Bayesianism is weakly Bayes plausible (it is even Bayes plausible), and hence the sender gains from sequential persuasion when the receiver uses this updating rule. 

Second, we consider the case where the sender's payoff is linear in the belief of the receiver. Such environments are trivial if the receiver is Bayesian since by Bayes plausibility the payoff to the sender is constant in the information they provide. However, in Proposition \ref{prop-transparent-linear} we show that the sender gains from sequential persuasion in every such environment whenever the receiver's updating rule is interior, satisfies eventual learning, and is `balanced'. Balancedness restricts the directions of the bias (i.e., the directions from the Bayesian posteriors to the biased posteriors) --- see Definition \ref{def-balanced}; conservative Bayesianism is a balanced updating rule.  

Proposition \ref{prop-common-pref} concerns the case where sender and receiver have identical preferences. We give necessary and sufficient conditions for the sender to gain from sequential persuasion, assuming that the receiver's updating is interior and satisfies eventual learning. Finally, in Proposition \ref{prop-CS} we show that if preferences are as in the classic model of \citet{crawford1982strategic}, where preferences are misaligned by a constant bias term, the sender gains from sequential persuasion if the receiver's updating rule is interior, satisfies eventual learning, and is balanced.

To gain intuition for why sequential persuasion can outperform static persuasion, suppose that the receiver is a conservative Bayesian. There may be an action $a$ that yields a high payoff to the sender (either in a particular state or in all states), but for the receiver to choose this action they need to hold a relatively extreme belief. With only one signal, receiver's posterior cannot be too far from the prior due to the bias, so $a$ will not be chosen, but by repeating the information multiple times receiver's belief can drift further away to a point where $a$ is selected. Second, even if a belief that induces $a$ can be achieved with a single signal, the sender may be able to increase the probability that $a$ is chosen by providing additional information when the initial signal realization induces other actions.\footnote{To illustrate, in the `Prosecutor-Judge' example, even if the first realized signal indicates that the defendant is innocent, the belief of the judge remains interior, and the prosecutor can induce conviction with positive probability by providing additional information.} Third, with a single experiment the distribution over receiver's biased posteriors is Bayes plausible, but with multiple rounds this is no longer the case due to the disagreement that forms between the sender and receiver after the first signal. This implies that new action distributions that are preferable to the sender may be achieved. Propositions \ref{prop-transparent-2-actions}-\ref{prop-CS} exemplify these forces in economically relevant environments.

\subsection{Related literature}

The two recent papers by \citet{kobayashi2025dynamic} and \citet[Section 4.2.3]{yang2026stochastic} ask a similar question to the one we study here, but instead of looking at specific classes of persuasion environments, their results characterize updating rules that do not allow the sender to gain from sequential persuasion in any environment.\footnote{These two papers were developed simultaneously and independently of our work. The first draft of our paper was publicly posted on arXiv before these papers were available.} The results in these two papers differ in several nuanced ways, but both show that the key property of an updating rule that makes sequential persuasion irrelevant is \emph{divisibility}, a property originally defined in \citet{cripps2018divisible}. The proofs in these papers are based on abstract duality arguments and hence do not speak to the mechanisms that make sequential persuasion outperform static persuasion, or to whether this is an issue in economically relevant environments. In contrast, our paper focuses on some of the classic environments considered in the literature, and we show how properties of the updating bias allow the sender to gain from dynamic persuasion in each of these.

Another closely related paper is \citet{levy2022persuasion}, in which the receiver suffers from correlation neglect when obtaining multiple signals. The main result shows that the sender can take advantage of this bias to obtain any state-contingent posterior distribution. In contrast, a sender interacting with a biased receiver as defined in this paper has limited persuasion power even with an arbitrarily long sequence of experiments. See the discussion in Section \ref{sec-discussion} for a proof of this claim, as well as for more on the relation between the biases considered in this paper and correlation neglect.

More broadly, this paper is related to papers that introduce frictions to the receiver's interpretation and processing of information. As already discussed above, our work builds on the results in \cite{de2022non}. In \cite{augias2020persuading} the receiver is a wishful thinker who biases their belief (at a cost) whenever this increases their expected utility. \cite{galperti2019persuasion} studies a BP problem with a receiver that ex ante assigns zero probability to some states and updates in a non-Bayesian manner upon observing unexpected news. In \cite{eliaz2021persuasion} the sender adds ambiguity to their message and as a result the receiver fails to properly apply Bayes rule. Finally, \cite{bloedel2018persuasion}, \cite{wei2021persuasion}, and \cite{dall2024persuading} assume that the receiver is rationally inattentive and incurs a cost when processing the signal. All these papers consider static persuasion, while our focus is on the potential benefits of sequential information provision under non-Bayesian updating.

Another strand of related literature contains papers in which information revelation is sequential; this is by now a large literature and we only mention a few of the more closely related works. In \cite{su2021bayesian} and \cite{ni2023sequential}, the sender is restricted to a limited set of experiments and can therefore benefit from sending multiple signals. \cite{zhang2024receiver} studies a model in which the receiver is rationally inattentive, showing that sequential persuasion can endogenously arise in this case. \cite{le2019persuasion} and \cite{escude2023slow} consider a sender who sequentially provides information to a receiver subject to an exogenous constraint on the rate of information transmission. Several papers \citep{li2021sequential, koessler2022long, wu2023sequential} study the equilibrium of BP games with multiple senders who move sequentially. These papers and the rest of the literature restrict attention to Bayesian agents.


\section{The model}\label{sec-model}

\subsection{Environment}\label{subsec-environment}

We start with some notation. For any set $Y$ let $\Delta(Y)$ be the set of finite-support probability distributions on $Y$. If $Y$ is a set in some linear space then $Conv(Y)$ denotes the convex hull of $Y$. Let $\Omega$ be the finite set of states. Typical elements of $\Delta(\Omega)$ are denoted by $\mu$ or $\nu$ and are viewed as vectors in $\RR^\Omega$ with non-negative coordinates that sum up to 1. The support of $\mu$ is the set $Supp(\mu)=\{\omega\in \Omega :~ \mu(\omega)>0\}$, and we let $\Delta^o(\Omega)=\{\mu\in \Delta(\Omega):~ Supp(\mu)=\Omega\}$ be the set of full-support beliefs. Elements of $\Delta(\Delta(\Omega))$ will usually be denoted by $\tau$ and their support by $T=Supp(\tau)$; we say that $\tau$ is interior if $T\subseteq \Delta^o(\Omega)$. 

As in \citet{KamenicaGentzkow2011}, we consider a game between a sender called S and a receiver called R: S moves first choosing an information structure, a signal is realized with a distribution that depends on the true state $\omega\in \Omega$ and on the choice of S, and finally R chooses an action $a$ from some set $A$. The utility function for R is $u(\omega,a)$ and for S is $v(\omega,a)$. The action that maximizes the expected utility of R at belief $\mu$ is denoted by $\hat a(\mu)$, and as usual we assume that a maximizer exists and that R chooses an S-optimal action when indifferent. Thus, the expected payoff of S given that both R and S hold belief $\mu$ is $\hat{v}(\mu) := \EE_{\omega\sim\mu}[v(\omega, \hat{a}(\mu))]$. 

Throughout we assume that S and R start with a full-support common prior $\mu_0$. Let $X$ be the set of signals and let $\Sigma_X$ be the set of information structures with signals in $X$ (i.e., mappings from $\Omega$ to $\Delta(X)$). Using Bayes rule, any $\sigma\in \Sigma_X$ induces a distribution over posteriors $\tau_\sigma$ in the usual way: A belief $\mu$ is in the support $T_\sigma$ if there is a positive $\sigma$-probability signal $x\in X$ that induces posterior $\mu$, and the probability of any such $\mu$ is the total $\sigma$-probability of all the signals $x$ that induce it.    

As was shown in \cite{KamenicaGentzkow2011}, a posterior distribution $\tau$ can be induced by some $\sigma$ if and only if it is Bayes plausible, and hence the maximal equilibrium expected utility for S in this game is given by the concave closure of $\hat{v}$ at $\mu_0$:
\begin{eqnarray}\label{eqn-V}
     V(\mu_0) &:=& \max_{\tau}~~ \sum_{\mu\in T} \tau(\mu)\hat v(\mu) \\
     &&\textit{s.t. } \sum_{\mu\in T} \tau(\mu)\mu = \mu_0. \nonumber
\end{eqnarray}

\subsection{Biased updating}\label{subsec-bias}

We deviate from the standard theory by relaxing the assumption that R updates beliefs in a Bayesian way. Instead, we assume that R systematically distorts posterior beliefs as in \cite{de2022non}: There is a `bias function' $D:\Delta^o(\Omega)\times \Delta(\Omega)\to \Delta(\Omega)$, such that $D(\nu,\nu')$ is the belief that R holds given prior $\nu$ and Bayesian posterior $\nu'$.

Notice that $D$ is defined only for full-support priors. This is sufficient for our purposes because throughout the paper we assume that the belief of R remains interior. We formalize this in the following definition.

\begin{definition}\label{def-interior}
The bias $D$ is \textbf{interior} if $Supp(D(\nu,\nu'))=\Omega$ for every prior $\nu$ and Bayesian posterior $\nu'$.
\end{definition}

When $D$ is interior R never completely rules out any state, regardless of the information they receive. This is the key property that distinguishes our R from a Bayesian receiver. We note that any $D$ can be approximated by an interior one simply by mixing it with some fixed interior belief with an arbitrarily small weight. 

For our results below we will need to further assume that by repeatedly sending information that induces a fixed Bayesian posterior the belief of R converges to that posterior. To formalize this property we introduce the following notation. Given $\nu\in\Delta^o(\Omega)$ and $\nu'\in\Delta(\Omega)$, let $D^1(\nu,\nu')=D(\nu,\nu')$, and for every $n>1$ let $D^n(\nu,\nu') = D(D^{n-1}(\nu,\nu'),\nu')$. Thus, $D^n(\nu,\nu')$ is R's belief when starting at the prior $\nu$ and after receiving $n$ signals that repeatedly induce a Bayesian posterior $\nu'$.

\begin{definition}\label{def-eventual-learning}
(i) The bias $D$ satisfies \textbf{eventual learning} if $\lim_n D^n(\nu,\nu')=\nu'$ for every pair $\nu\in\Delta^o(\Omega)$ and $\nu'\in\Delta(\Omega)$.\\
(ii) The bias $D$ satisfies \textbf{eventual learning of states} if $\lim_n D^n(\nu,\delta_\omega)=\delta_\omega$ for every $\nu\in\Delta^o(\Omega)$ and $\omega\in\Omega$.
\end{definition}

Note that eventual learning of states is a weaker property than eventual learning --- for most of our results it will be sufficient to assume that repeatedly sending fully revealing signals induces R's belief to converge to the vertex of $\Delta(\Omega)$ corresponding to the realized state. 

\medskip

When R has bias $D$, any realization of a signal $x$ from an information structure $\sigma$ induces a pair of posteriors: S holds the Bayesian posterior $\mu$, while R holds the biased posterior $D(\mu_0,\mu)$. Thus, any $\sigma$ now induces a distribution $\lambda_\sigma^D$ over $\Delta(\Omega)\times \Delta(\Omega)$, where the first coordinate corresponds to the Bayesian belief of S and the second to the biased belief of R. The marginal of $\lambda_\sigma^D$ on its first coordinate is $\tau_\sigma$ as defined earlier, while the marginal on the second coordinate is denoted $\tau_\sigma^D$.  

Finally, we recall the characterization of the highest achievable payoff for S with one-shot persuasion from \citet{de2022non}. Given bias $D$, the maximization problem for S can be written in a similar way to (\ref{eqn-V}) as follows:
\begin{eqnarray}\label{eqn-V-biased}
     V_D(\mu_0) &:=& \max_{\tau}~~ \sum_{\mu\in T} \tau(\mu) \EE_{\omega\sim \mu}[v(\omega,\hat a(D(\mu_0,\mu)))]\\
     &&\textit{s.t. } \sum_{\mu\in T} \tau(\mu)\mu = \mu_0. \nonumber
\end{eqnarray}
Indeed, when S holds posterior $\mu$ the belief of R is $D(\mu_0,\mu)$, so they would choose the action $\hat a(D(\mu_0,\mu))$ and the expected payoff for S would be $\EE_{\omega\sim \mu}[v(\omega,\hat a(D(\mu_0,\mu)))]$. Therefore, (\ref{eqn-V-biased}) expresses the highest possible expected payoff for S over all Bayes plausible posterior distributions. We denote $\hat v_D(\mu) = \EE_{\omega\sim \mu}[v(\omega,\hat a(D(\mu_0,\mu)))]$, so that (\ref{eqn-V-biased}) becomes just like a standard persuasion problem with the function $\hat v$ replaced by $\hat v_D$. 

\subsubsection{Conservative Bayesianism}

To illustrate some of our results we will use a particular type of bias function, denoted $D_\alpha$, which describes conservative Bayesianism: There is a bias parameter $\alpha\in (0,1)$, so that given a prior belief $\nu$ and a Bayesian posterior $\nu'$ the belief is
\begin{equation}\label{eqn-bias-CB}
    D_\alpha(\nu,\nu') = \alpha \nu+(1-\alpha)\nu'.
\end{equation}
It is clear that $D_\alpha$ is interior and satisfies eventual learning. We point out a couple of additional properties it satisfies: First, given a prior $\nu$ and $\alpha$, for a biased agent to hold some belief $\nu''$ after updating there must exist a belief $\nu'$ such that $\nu''=\alpha \nu+(1-\alpha)\nu'$. Equivalently, we need that $\nu'=\frac{1}{1-\alpha}[\nu''-\alpha\nu]\in \Delta(\Omega)$, which is satisfied if and only if $\nu''(\omega)\ge\alpha\nu(\omega)$ for every $\omega\in \Omega$.

Second, if $\tau$ is any posterior distribution that satisfies the Bayes plausibility constraint $\sum_{\nu'\in T} \tau(\nu')\nu' = \nu$, then after applying the bias to each belief $\nu'\in T$ the resulting posterior distribution would still satisfy this constraint. The next lemma immediately follows from these properties.

\begin{lemma}\label{lemma-feasible}
    Fix a prior $\nu$ and suppose that R is a conservative Bayesian with bias $D_\alpha$. A (biased) posterior distribution $\tau$ can be induced by some information structure $\sigma$ if and only if (i) $\sum_{\nu''\in T} \tau(\nu'')\nu'' = \nu$, and (ii) $\nu''\ge \alpha \nu$ for every $\nu''\in T$.
\end{lemma}

\subsection{Sequential persuasion}\label{subsec-sequential}

We expand the set of strategies available to S to allow for sequential persuasion. A \emph{sequential persuasion strategy} $\sigma$ specifies either the information structure from $\Sigma_X$ to be used or the termination of the persuasion process as a function of the history of realized signals. We require that persuasion terminates in finite time, although the length of the persuasion process may depend on the realization of signals. Figure \ref{fig-seq-strategy} illustrates a hypothetical sequential persuasion strategy that terminates after a maximum of four periods.

\begin{figure}
\centering
\begin{tikzpicture}[scale = 0.7,
  level 1/.style={sibling distance=70mm, level distance=12mm},
  level 2/.style={sibling distance=50mm, level distance=15mm},
  level 3/.style={sibling distance=30mm, level distance=15mm},
  every node/.style={align=center, minimum size=0.5cm, font=\footnotesize}
  ]
  \node {\small\textit{S begins}}
    child {node {$x$}
      child {node {$x$} child {node {$STOP$}}}
      child {node {$x'$}
        child {node {$x$} child {node {$STOP$}}}
        child {node {$x'$} child {node {$STOP$}}}
      }
    }
    child {node {$x'$} child {node {$STOP$}}};
\end{tikzpicture} 
\caption{An illustration of a sequential persuasion strategy. The set of possible signals is $X=\{x,x'\}$. If $x'$ is realized in the first period then the strategy terminates, and if $x$ is realized then another information structure is selected; if $x$ is realized again the strategy terminates, and if $x'$ is realized then a third and final information structure is used. Note that the distribution over $X$ (conditional on each $\omega$) may differ after different histories.}
 \label{fig-seq-strategy}
\end{figure}

When information arrives sequentially, we need to specify how R updates their belief to reach a final posterior. For most of the paper we focus on an updating procedure in which the bias is applied after each signal, and the belief at the end of the previous period serves as the prior both for the purpose of calculating the new Bayesian posterior and as the first argument in the bias $D$.\footnote{In subsection \ref{subsec-other-updating-procedures} we briefly discuss alternative updating procedures.} Formally, let $\mu^\sigma(x_1,\ldots,x_n)$ be R's belief after $(x_1,\ldots,x_n)$ are observed and given that S uses strategy $\sigma$. After the first signal we have $\mu^\sigma(x_1)=D(\mu_0, \mu)$, where $\mu$ is the Bayesian posterior calculated from the prior $\mu_0$ and the realized signal $x_1$ (and given $\sigma$). For any subsequent signal we define the belief recursively by $\mu^\sigma(x_1,\ldots,x_n)=D(\mu^\sigma(x_1,\ldots,x_{n-1}), \mu)$, where $\mu$ is the Bayesian posterior calculated from the prior $\mu^\sigma(x_1,\ldots,x_{n-1})$ and the realized signal $x_n$.

Just as in the case of one-shot persuasion, any sequential persuasion strategy $\sigma$ induces a joint distribution over the posteriors of S and R, which we again denote by\footnote{We emphasize that to derive $\lambda_\sigma^D$ the probability of any sequence of signals $(x_1,\ldots,x_n)$ is calculated as $\sum_\omega \mu_0(\omega)\prod_{k=1}^n\sigma_{x_1,\ldots,x_{k-1}}(x_k|\omega)$, where $\sigma_{x_1,\ldots,x_{k-1}}$ denotes the information structure used conditional on this history of signals. In other words, R's bias does not affect the likelihood of signal realizations, only the induced posterior.} $\lambda_\sigma^D\in\Delta(\Delta(\Omega)\times \Delta(\Omega))$. As before, the marginal of $\lambda_\sigma^D$ on the first coordinate is $\tau_\sigma$ --- the Bayesian posterior distribution for S, and the marginal on the second coordinate is $\tau_\sigma^D$ --- the biased posterior distribution for R. Thus, the payoff to S from any strategy $\sigma$ is given by
\begin{equation}\label{eqn-payoff-sequential}
    \int \EE_{\omega\sim\mu}[v(\omega, \hat a(\mu'))] d\lambda_\sigma^D(\mu,\mu').
\end{equation}

The following example illustrates the updating procedure and the resulting joint posterior distribution in the case where R is a conservative Bayesian.

\begin{example}\label{example-PbP-N}
Suppose that $\Omega=\{\omega_0,\omega_1\}$ and let $\mu=\mu(\omega_1)\in [0,1]$ be the belief that the state is $\omega_1$. Let the prior be $\mu_0=0.5$ and fix the bias parameter $\alpha=0.5$. Consider a 2-period sequential persuasion strategy $\sigma$ such that in each period the signal is either $x^0$ or $x^1$ and the probability that the signal matches the state is $\beta\in [0.5,1]$. Note that in $\sigma$ the information structure in the second period does not depend on the realized signal in the first period.

Suppose first that the signal $x^1$ is realized in the first period. The Bayesian posterior after this signal is $\mu=\beta$, so the biased posterior after the first signal is $\mu^\sigma(x^1)=0.5*\mu_0+ 0.5*\beta=0.25+0.5\beta$. The latter serves as the prior for the next period. Thus, if $x^1$ is observed again then the Bayesian posterior after the second signal is $\mu=\frac{(0.25+0.5\beta)*\beta}{(0.25+0.5\beta)*\beta+(0.75-0.5\beta)*(1-\beta)}$, and the final biased posterior is given by 
$$\mu^\sigma(x^1,x^1)=0.5*\mu^\sigma(x^1)+0.5*\frac{(0.25+0.5\beta)*\beta}{(0.25+0.5\beta)*\beta+(0.75-0.5\beta)*(1-\beta)}.$$
If instead $x^0$ is observed in the second period then the Bayesian posterior is $\mu=\frac{(0.25+0.5\beta)*(1-\beta)}{(0.25+0.5\beta)*(1-\beta)+(0.75-0.5\beta)*\beta}$, and the final biased posterior is given by 
$$\mu^\sigma(x^1,x^0)=0.5*\mu^\sigma(x^1)+0.5*\frac{(0.25+0.5\beta)*(1-\beta)}{(0.25+0.5\beta)*(1-\beta)+(0.75-0.5\beta)*\beta}.$$
Similar calculations yield the posteriors $\mu^\sigma(x^0,x^0)$ and $\mu^\sigma(x^0,x^1)$. 

To complete the description of the biased posterior distribution $\tau_\sigma^D$, we also need to calculate the likelihood of each of these four possible posteriors. Note that this calculation is made from the point of view of S and is not affected by the bias. Thus, the probability that $(x^1,x^1)$ is realized is $0.5*\beta^2+0.5*(1-\beta)^2$, the probability that $(x^1,x^0)$ is realized is $\beta(1-\beta)$, etc. See Figure \ref{fig-biased-unbiased-posterior} for the induced posterior distributions $\tau_\sigma$ and $\tau_\sigma^D$ when $\beta = 0.75$.
\hfill \qed
\end{example}

\begin{figure}[ht]
\centering
\begin{tikzpicture}[scale=0.65]
\begin{axis}[
    ybar,
    bar width=5pt,
    width=20cm,
    height=10cm,
    enlargelimits=0,
    ymin=0,
    ymax=0.5,
    xmin=0,
    xmax=1,
    xlabel={\Large Posterior belief for $\omega_1$},
    ylabel={\Large Probability of posterior},
    xtick={0.1,0.27,0.45,0.5,0.55,0.73,0.9,1},
    xticklabels={0.1\\$x^0x^0$,0.27\\$x^0x^0$,0.49\\$x^0x^1$,0.5\\$x^0x^1$\\/$x^1x^0$,0.51\\$x^1x^0$,0.73\\$x^1x^1$,0.9\\$x^1x^1$,1},
    xticklabel style={align=center, font=\small},
    tick label style={font=\small},
    bar shift=0pt,
    nodes near coords,
    legend style={at={(0.5,-0.30)}, anchor=north, legend columns=-1}
]

\addplot+[ybar, fill=blue!80] 
  coordinates {(0.9,0.3125) (0.5,0.375) (0.1,0.3125)};   

\addplot+[ybar, fill=red!40] 
   coordinates {(0.2708,0.3125) (0.55,0.1875) (0.7292,0.3125) (0.45,0.1875)};
   
\legend{Bayesian(S), Biased(R)}

\end{axis}
\end{tikzpicture}
\caption{The posterior distributions $\tau_\sigma^D$ (red) and $\tau_\sigma$ (blue) for the 2-period sequential persuasion strategy of Example \ref{example-PbP-N} when $\beta=0.75$. By considering the pair of signals that induce each posterior, one can also visualize the joint distribution $\lambda_\sigma^D$. Note that the sequences $(x^0,x^1)$ and $(x^1,x^0)$ induce the same Bayesian posterior but different biased posteriors.}
\label{fig-biased-unbiased-posterior}
\end{figure}

\subsection{Optimal versus beneficial sequential persuasion}

In the static case of one-shot persuasion, the key to the applicability of the concavification method (\ref{eqn-V-biased}) is that the Bayesian posterior $\mu$ held by S uniquely pins down the biased belief $D(\mu_0,\mu)$ held by R, and hence also pins down the action that R takes at this belief. In contrast, with sequential persuasion this is no longer true, as there can be different sequences of signals that lead to the same Bayesian posterior for S but to different biased posteriors for R. We have already seen this in Example \ref{example-PbP-N} above: The sequences of signals $(x^1,x^0)$ and $(x^0,x^1)$ both result in the same Bayesian posterior of $0.5$, but R's beliefs differ after these two realizations (see Figure \ref{fig-biased-unbiased-posterior}).\footnote{\citet{hoffman2011simultaneous} experimentally show that subjects' posteriors often change depending on the order in which signals are presented.} If R's optimal action at the posterior induced by $(x^1,x^0)$ is different from the one at the posterior induced by $(x^0,x^1)$, then we would not be able to express the payoff of S (\ref{eqn-payoff-sequential}) in the same way as in (\ref{eqn-V-biased}), and hence would not be able to neatly characterize the optimal payoff.

The failure of systematic belief distortion makes it challenging to find optimal sequential persuasion strategies even in simple environments.\footnote{See Subsection \ref{subsec-first-best} for more on the difficulty of solving for an optimal sequential persuasion strategy.} Instead, we demonstrate the relevance of sequential persuasion by proving that in some of the workhorse models considered in the literature the ability to send multiple signals is beneficial to S. Furthermore, we explore the different mechanisms that make sequential persuasion outperform static persuasion and relate them to properties of the bias $D$. 

For a given persuasion environment, prior $\mu_0$, and bias $D$, we say that \emph{S gains from sequential persuasion} if there exists a sequential persuasion strategy $\sigma$ that yields a strictly higher payoff for S than $V_D(\mu_0)$ given in (\ref{eqn-V-biased}). 


\section{Two actions}\label{sec-2-actions}

We start with the class of persuasion environments that have received the most attention in the literature: R has only two actions to choose from and S has `transparent motives' \citep{lipnowski2020cheap}, i.e., S's preferences are state-independent.\footnote{When S has transparent motives we write $v(a)$ for the utility of S when R takes action $a$.} More precisely, we consider the following environments.

\begin{definition}\label{def-2-actions}
In a \textbf{2-action persuasion environment}, the set of actions for R is $A=\{a,a'\}$ and S has transparent motives with $v(a)>v(a')$. The prior $\mu_0$ satisfies $\hat a(\mu_0)=a'$, and there is some interior $\mu$ such that $\hat a(\mu)=a$.
\end{definition}

To understand how sequential persuasion can be beneficial in 2-action environments, as well as to gain some intuition for the general result that follows, we start with the classic Prosecutor-Judge example of \citet{KamenicaGentzkow2011}.

\begin{example}\label{example-2-actions}
Let $\Omega=\{I,G\}$ and denote by $\mu=\mu(G)$ the probability of state $G$. Suppose that $\hat{a}(\mu)=a$ for $\mu\ge 0.5$ and $\hat{a}(\mu)=a'$ otherwise. The prior is $\mu_0=0.25$. Assume that R is a conservative Bayesian with $\alpha=0.5$. 

Following the same logic as in the standard Bayesian case, it is easy to see that the optimal one-shot persuasion strategy is to induce Bayesian posteriors of $\mu_1=0$ (with probability 2/3) and $\mu_2=0.75$ (with probability 1/3). This induces biased posteriors of $D(\mu_0,\mu_1)=0.125$ and $D(\mu_0,\mu_2)=0.5$, so the probability of R choosing action $a$ (S's favorite) is 1/3.

However, in case the realized Bayesian posterior turns out to be $\mu_1$ (which in particular implies that the state is $I$), the fact that R's belief remains interior opens up the possibility to continue the persuasion process. Specifically, suppose that in this case S repeatedly uses binary information structures with signals $\{i,g\}$ such that the signal matches the state with probability $1-\epsilon$ for some small $\epsilon>0$. Since the state is $I$, in each period the signal $g$ is realized with probability $\epsilon$, in which case the belief of R drifts towards the $G$ state. After sufficiently many realizations of $g$ R's belief will cross the 0.5 threshold leading them to choose $a$. Therefore, the overall probability of $a$ is greater than in the optimal one-shot persuasion strategy. 
\hfill \qed
\end{example}

To generalize the above example, we define an additional property of the bias $D$ that is a weakening of classic Bayes plausibility. This property has been used in \citet{cripps2018divisible} as one of the axioms that characterize Bayesian updating, but here it turns out to be necessary and sufficient for S to gain from sequential persuasion in every 2-action environment.

\begin{definition}\label{def-convex-hull}
The bias $D$ is \textbf{weakly Bayes plausible} if for every prior $\nu$ and every finite collection $\{\nu'_1,\ldots,\nu'_K\}$ such that $\nu\in Conv(\{\nu'_1,\ldots,\nu'_K\})$, we have that 
$$\nu\in Conv(\{D(\nu,\nu'_1),\ldots,D(\nu,\nu'_K)\}).$$
\end{definition}

Note that conservative Bayesianism is weakly Bayes plausible since it is Bayes plausible (recall Lemma \ref{lemma-feasible}). It is not hard to construct biases that satisfy only the weak version --- as a simple example take a generalized version of conservative Bayesianism in which the parameter $\alpha$ changes in an arbitrary way with the pair $\nu,\nu'$. Our first main result follows.


\begin{proposition}\label{prop-transparent-2-actions}
Suppose that $D$ is interior and satisfies eventual learning. Then S gains from sequential persuasion in every 2-action persuasion environment if and only if $D$ is weakly Bayes plausible. 
\end{proposition}


\section{Linear preferences}\label{sec-linear}

The next class of persuasion environments we consider is the following.

\begin{definition}\label{def-linear}
In a \textbf{linear persuasion environment}, the state space is $\Omega\subseteq \RR^k$ for some $k\ge 1$, and $Conv(\Omega)$ is a set of full dimension. The set of actions is $A=\RR^k$, and $\hat a(\mu) = \EE_{\omega\sim\mu} [\omega]$ for every $\mu\in \Delta(\Omega)$. S has transparent motives, and $v(a)=\beta\cdot a$ for some non-zero vector $\beta\in \RR^k$. 
\end{definition}

Thus, in a linear environment states are real vectors and the optimal action for R is the mean of their belief; such receiver preferences have been used extensively in the literature \citep[e.g.,][]{crawford1982strategic}. The sender has state-independent linear preferences over the actions of the receiver. Notice that if R is Bayesian then persuasion is pointless --- when S induces posterior distribution $\tau$ their payoff is 
$$\sum_{\mu\in T}\tau(\mu)(\EE_{\omega\sim\mu} [\omega]\cdot \beta) = \beta\cdot \left(\sum_{\mu\in T}\tau(\mu)\EE_{\omega\sim\mu} [\omega]\right) = \beta \cdot \EE_{\omega\sim\mu_0} [\omega],$$
where the last equality follows from Bayes plausibility. Hence, the payoff of S is constant in the information they provide. 

Furthermore, when R is a conservative Bayesian their posterior distribution is also Bayes plausible (Lemma \ref{lemma-feasible}), so no one-shot persuasion strategy would benefit S. However, with sequential persuasion S and R typically hold different beliefs after the first period, which may allow S to achieve action distributions that were not feasible in the Bayesian case and potentially increase their payoff. The following example illustrates this point.

\begin{example}\label{example-linear}
Consider the following linear persuasion environment. The states are $\Omega=\{0,1\}\subseteq \RR$ and the payoff for S is $v(a)=a$. Let $\mu$ denote the belief that the state is $\omega=1$ and let the prior be $\mu_0=0.5$. Suppose that R is a conservative Bayesian with $\alpha=0.5$. 

As pointed out above, with a single period S cannot benefit from persuasion so their optimal payoff is 0.5. Consider the following 2-period sequential strategy: In the first period the information structure fully reveals the state; in the second period, if the state was revealed to be $\omega=0$ then persuasion terminates; and if the state was revealed to be $\omega=1$ then another fully revealing signal is sent.

Under this persuasion strategy R's posterior is $\mu_1=1/4$ with probability 0.5 (when the state is $\omega=0$) and $\mu_2=7/8$ with probability 0.5 (when the state is $\omega=1$). Thus, the payoff for S is $0.5*1/4+0.5*7/8>0.5$.\footnote{Of course, S can continue to send fully revealing signals when the state is $\omega=1$, and the posterior of R would converge to $\mu=1$. The limit payoff for S is $0.5*1/4+0.5*1 = 5/8$. It turns out that this is not the optimal sequential persuasion strategy and S can achieve an even higher payoff --- see Subsection \ref{subsec-first-best} for more on this example.}
\hfill \qed
\end{example}

As the above example shows, gaining from sequential persuasion when R is a conservative Bayesian is relatively easy because one-shot persuasion has no merit. For more general bias functions $D$, we should also consider the case where optimal one-shot persuasion is non-trivial. The next property takes care of such cases.

\begin{definition}\label{def-balanced}
The bias $D$ is \textbf{balanced} if for every prior $\nu$ and every finite collection $\{\nu_1',\ldots,\nu'_K\}$ such that $\nu\in Conv(\{\nu_1',\ldots,\nu'_K\})$ we have that 
$$\underline{0} \in Conv(\{D(\nu,\nu_1')-\nu'_1, \ldots, D(\nu,\nu_K')- \nu'_K\}),$$
where $\underline{0}$ is the zero vector in $\RR^\Omega$. 
\end{definition}

To understand the above definition, note that the difference $D(\nu,\nu_k')-\nu'_k\in \RR^\Omega$ is a vector that points in the direction from the Bayesian posterior $\nu'_k$ to the biased one $D(\nu,\nu_k')$, i.e., it is the direction of the bias given prior $\nu$ and posterior $\nu'_k$. When $D$ is balanced the directions of the bias resulting from the various posteriors of an information structure are not one-sided --- they are sufficiently spread out so that one cannot strictly include all of them in a half-space through the origin. For example, in the case of two posteriors that contain the prior in their convex hull, the condition requires that the directions of the two biases are exactly opposite. Note that conservative Bayesianism is balanced, as is any generalized version of it in which the weight $\alpha$ varies with the prior $\nu$ and the Bayesian posterior $\nu'$. 

We can now state the result of this section.

\begin{proposition}\label{prop-transparent-linear}
Suppose that $D$ is interior, satisfies eventual learning of states, and is balanced. Then S gains from sequential persuasion in every linear persuasion environment.
\end{proposition}


\section{Common preferences}

We next consider the case of a benevolent S whose preferences are identical to those of R. The following proposition characterizes environments in which S gains from sequential persuasion under our minimal assumptions on the bias $D$. 

\begin{proposition}\label{prop-common-pref}
Suppose that S and R have the same preferences, $u(\omega,a)=v(\omega,a)$ for all $\omega$ and $a$. Assume that $D$ is interior and satisfies eventual learning of states. Then S gains from sequential persuasion if and only if 
\begin{equation}\label{eqn-common-pref}
    V_D(\mu_0) < \sum_{\omega\in \Omega} \mu_0(\omega)\max_a\{v(\omega,a)\}.
\end{equation}
\end{proposition}

The right-hand side of (\ref{eqn-common-pref}) is the expected payoff for S (and R) under full information. Since we assume $D$ satisfies eventual learning of states, S can achieve a payoff arbitrarily close to full information by repeatedly sending fully revealing signals. This implies that (\ref{eqn-common-pref}) is sufficient for S to gain from sequential persuasion. Conversely, since S and R have the same preferences, it is straightforward to show that with any persuasion strategy the payoff to S is weakly higher when R is Bayesian than when R is biased (this does not require any assumption on $D$). Hence, the full-information payoff is an upper bound for the payoff from any sequential persuasion strategy, so (\ref{eqn-common-pref}) is also necessary. 

While Proposition \ref{prop-common-pref} provides a characterization, checking whether (\ref{eqn-common-pref}) holds requires solving for $V_D(\mu_0)$, which may not be easy.\footnote{Note that fully revealing the state is not necessarily the optimal one-shot strategy for S. In fact, it follows from a result of \citet{whitmeyer2024blackwell} that any distortion of Bayes rule (within the class of biases considered by \citet{de2022non}) causes violations of monotonicity with respect to the Blackwell order in some decision problem. Thus, S and R will sometimes disagree on which of two information structures is more valuable.} The next corollary provides a simple sufficient condition for S to gain from sequential persuasion. 

\begin{corollary}\label{coro-common-pref}
Under the assumptions of Proposition \ref{prop-common-pref}, if there exists a state $\bar\omega$ such that for all $\mu\in \Delta(\Omega)$
\begin{equation}\label{eqn-empty-intersect}
    \argmax_a u(\bar\omega,a) \bigcap \argmax_a \EE_{\omega\sim D(\mu_0,\mu)}[u(\omega,a)] = \emptyset,
\end{equation}
then S gains from sequential persuasion. 
\end{corollary}

We illustrate Corollary \ref{coro-common-pref} with the following example.

\begin{example}\label{example-common-pref}
   Let $\Omega=\{\omega_0,\omega_1\}$ and identify beliefs with $\mu=\mu(\omega_1)\in[0,1]$. Suppose the prior is $\mu_0=0.5$. R can choose from three actions, $\{a_0,a_1,b\}$, with payoffs given by $u(\omega_0,b)=u(\omega_1,b)=0$, $u(\omega_0,a_0)=u(\omega_1,a_1)=1$, $u(\omega_0,a_1)=u(\omega_1,a_0)=-2$. Thus, $a_0$ is optimal for $\mu\in[0,1/3]$, $a_1$ is optimal for $\mu\in[2/3,1]$, and $b$ is optimal for intermediate beliefs $\mu\in[1/3,2/3]$. Recall that $v$ is identical to $u$.

   Suppose that R is a conservative Bayesian with parameter $\alpha$. Then the set of feasible biased posteriors that S can induce with one period of persuasion is the interval $[\alpha/2, 1-\alpha/2]$. If $\alpha>2/3$ then this is not sufficient to convince R to choose any of the extreme actions $a_0$ or $a_1$, implying that condition (\ref{eqn-empty-intersect}) holds. By repeatedly sending information showing the realized state R will eventually become sufficiently confident to choose these actions, so S would gain from sequential persuasion in this case. If on the other hand $\alpha\le 2/3$ then fully revealing the state in the first period already generates the maximal possible payoff, so sequential persuasion has no additional value. In fact, if R is a conservative Bayesian and $\Omega$ contains two states then condition (\ref{eqn-empty-intersect}) is also necessary for S to gain from sequential persuasion.  
\end{example}


\section{Crawford-Sobel preferences}

Our last result concerns the classic environment analyzed by \citet{crawford1982strategic}.

\begin{definition}\label{def-CS}
In a \textbf{Crawford-Sobel environment}, the state space is $\Omega\subseteq \RR$ and the set of actions is $A=\RR$. The payoff functions are $u(\omega,a)=-(a-\omega)^2$ and $v(\omega,a)=-(a-(\omega+b))^2$ for some $b\in \RR$.
\end{definition}

\begin{proposition}\label{prop-CS}
Suppose that $D$ is interior, satisfies eventual learning of states, and is balanced. Then S gains from sequential persuasion in every Crawford-Sobel environment.
\end{proposition}

In the proof of the proposition we show that, starting from an optimal one-shot persuasion strategy, S can further increase their payoff by appropriately choosing one of the possible posteriors and sending additional information if this posterior is realized. For example, if R is a conservative Bayesian then one can show that S's value function $\hat v_\alpha(\mu) = \EE_{\omega\sim \mu}[v(\omega,\hat a(D_\alpha(\mu_0,\mu)))]$ is convex, and hence that the optimal one-shot persuasion strategy is to fully reveal the state.\footnote{\citet[Section 6.2]{kamenica2009bayesian} consider the same environment without the bias and show that $\hat{v}(\mu)$ is convex.} If the bias $b$ is positive, and if the state turns out to be the highest, then S can increase their payoff by sending a second fully revealing signal that further pushes up R's posterior and induces a higher action; for a negative $b$ the same argument applies when the state is the lowest. For more general biases $D$ the function $\hat v_D(\mu)$ need not be convex; instead, the balancedness of $D$ guarantees that a similar argument can be applied.


\section{Discussion}\label{sec-discussion}

\subsection{Optimal sequential persuasion and the sender's first-best}\label{subsec-first-best}

One may wonder whether sequential persuasion is so powerful that it allows S to (almost) achieve their first-best payoff, meaning that when persuasion terminates R chooses S's favorite action in each state $\omega$. The answer is typically `no' --- there are limits to what S can achieve even with arbitrarily long persuasion. 

To understand why, consider again Example \ref{example-linear} from Section \ref{sec-linear}: There are two states $\Omega=\{0,1\}$, a belief $\mu$ specifies the probability of state $\omega=1$, the prior is uniform $\mu_0=0.5$, and the optimal action for R at belief $\mu$ is $a=\mu$. A well-known property of Bayesian updating is that for any information structure, conditional on the state being $\omega=0$, the expected posterior is (weakly) lower than the prior \citep{good1965list, francetich2014bayesian}. In other words, for a Bayesian receiver, conditional on $\omega=0$ the posterior process induced by any sequential persuasion strategy is a super-martingale. Since a conservative Bayesian mixes between the posterior and the prior with fixed weights, the same property holds also for them. This implies that in state $\omega=0$ the expected posterior of R at the end of any sequential persuasion strategy is at most 0.5. In particular, since the preferences of S are given by $v(a)=a$, it follows that the best payoff S can obtain in state $\omega=0$ is 0.5. Overall, the payoff of S is bounded above by 
$\mu_0*1+(1-\mu_0)*0.5 = 0.75$, strictly below the first-best payoff of 1.\footnote{In the Prosecutor-Judge example of Section \ref{sec-2-actions} (Example \ref{example-2-actions}), the same argument combined with Markov inequality implies that in state $I$ the probability of $a$ (conviction) is at most $0.5$, and hence that the overall probability of $a$ is at most $0.25*1+0.75*0.5=0.625$.}

Since the first-best is not feasible, it is natural to ask what is the optimal payoff that S can achieve. This turns out to be a difficult question even in this very basic environment with linear (state-independent) S preferences and R being a conservative Bayesian. As we argued in Example \ref{example-linear} above, S can achieve a payoff arbitrarily close to $0.5*1+0.5*0.25=0.625$ by initially sending a fully revealing signal, stopping if it turns out that $\omega=0$, and continuing to send fully revealing signals if it turns out that $\omega=1$. However, S can do better than that with the following strategy: In the first period use a binary information structure with signals $\{x,x'\}$ and conditional probabilities $\Pr(x|\omega=1)=1$, $\Pr(x|\omega=0)\approx 0.415$. If $x'$ is realized then persuasion terminates, and if $x$ is realized then send a fully revealing signal and follow the strategy previously described (stopping in state $\omega=0$ and continuing to send fully revealing signals in state $\omega=1$). With this strategy the payoff of S converges to 0.635.

More generally, an optimal sequential persuasion strategy is a solution to a dynamic programming problem that has a complicated structure even in simple environments. While S typically cannot achieve their first-best, our results show the usefulness of dynamic information design without finding the optimum.

\subsection{Comparison with correlation neglect}\label{subsec-neglect}

A well-documented bias in processing of multiple signals is correlation neglect. The paper by \citet{levy2022persuasion} studies BP when R suffers from this bias and shows that S can induce essentially any state-contingent posterior distribution by sending many signals and correlating them appropriately. In particular, S can achieve their first-best payoff.   

We claim that the type of biases considered in this paper are fundamentally different than correlation neglect. The most obvious difference is that the receiver in \citet{levy2022persuasion} is a standard Bayesian when only one signal is sent, while our R is biased, and this one-period bias fully determines their updating behavior with multiple signals. Put differently, if our R were Bayesian with one-shot persuasion, then they would be Bayesian with sequential persuasion as well.

Second, with correlation neglect the order in which signals are received is irrelevant; in fact, one can imagine that the receiver obtains all the signals at the same time, which is implicitly assumed in \citet{levy2022persuasion}. For our R the order of signals certainly matters, as illustrated in Example \ref{example-PbP-N}. Moreover, the sequential persuasion strategies we consider in our results take advantage of this ordering effect and make future information structures contingent on past realizations. 

Third, the result of \citet{levy2022persuasion} stands in stark contrast to our claim above that S typically cannot achieve their first-best payoff. This difference clearly indicates that these are separate models that generate diverging predictions in persuasion problems.

\subsection{Other sequential updating procedures}\label{subsec-other-updating-procedures}

Our view is that R's sequential updating procedure assumed thus far best captures the spirit of biased updating and of the fact that information arrives sequentially: The prior at the beginning of each period reflects the information accumulated in previous periods, and this prior is a sufficient statistic for calculating the next posterior. However, there are other natural ways to extend a biased updating rule from static to dynamic setups.\footnote{\citet{sarnoff2025structure} shows experimentally that for a large fraction of subjects the posterior of the previous period is not a sufficient statistic for current period updating.}

One example is when the bias is applied only after the persuasion process is over: R updates their belief in a Bayesian manner along the entire sequence of signals, and if the final Bayesian posterior is $\mu$, then their biased posterior is $D(\mu_0,\mu)$. This procedure essentially ignores the dynamic nature of sender's persuasion strategy and hence makes it irrelevant: Since with Bayesian updating the set of feasible posterior distributions under sequential persuasion is the same as under one-shot persuasion, it immediately follows that S never gains from sequential persuasion.

Another less trivial procedure is when the bias is applied in every period, but the anchor for the bias is the original prior $\mu_0$ rather than the last period posterior. Formally, $\mu_n(x_1,\ldots,x_n)=D(\mu_0,\mu)$, where $\mu$ is the Bayesian posterior calculated from the prior $\mu_{n-1}(x_1,\ldots,x_{n-1})$ and the realized signal $x_n$. Under this procedure it is still the case that different sequences of signal realizations may lead to the same Bayesian posterior but to different biased posteriors, so that beliefs are not systematically distorted and the concavification method is not applicable. Overall, this procedure assigns greater prominence to the initial belief $\mu_0$, which can make it harder for S to gain from sequential persuasion: As an example, conservative Bayesianism does not satisfy eventual learning since the posterior is pulled back towards the original prior after every period. A more careful analysis of this procedure is left for future work.


\appendix

\section{Proofs}


\bigskip

\noindent \textbf{Proof of Proposition \ref{prop-transparent-2-actions}} 

\smallskip

Before proving the proposition we first prove a simple lemma. Given any belief $\nu\in \Delta(\Omega)$ and $\epsilon>0$ let 
$$B(\nu,\epsilon)=\{\nu':~ \|\nu'-\nu\|<\epsilon\}$$
be the open ball around $\nu$ with radius $\epsilon$ (here $\|\cdot\|$ stands for the Euclidean distance). Also, given some sequential persuasion strategy $\sigma$ and $\omega\in\Omega$ we denote by $\tau_\sigma^D(\cdot|\omega)$ R's posterior distribution conditional on the true state being $\omega$.

\begin{lemma}\label{lemma-converge}
Suppose that $D$ is interior and satisfies eventual learning. Let $\nu_0\in \Delta^o(\Omega)$ be the prior and consider some target belief $\nu^*\in \Delta^o(\Omega)$. For every $\epsilon>0$ there exists a sequential persuasion strategy $\sigma$ such that $\tau_\sigma^D(B(\nu^*,\epsilon)|\omega)>0$ for every state $\omega$.    
\end{lemma}

\begin{proof}
  Given any belief $\nu\in\Delta^o(\Omega)$ consider an information structure $I_\nu$ with two possible signals $x,x'$ and where the probability of $x$ in state $\omega$ is $\frac{\nu^*(\omega)}{\nu(\omega)}\min_{\omega'\in\Omega}\left(\frac{\nu(\omega')}{\nu^*(\omega')}\right)$. Since both $\nu$ and $\nu^*$ are interior these probabilities are well-defined and strictly positive. Also, if $\nu$ is the prior and $x$ is the realized signal then the resulting Bayesian posterior is $\nu^*$. 

  Fix $\epsilon>0$ and $N\in \mathbf{N}$. Consider the sequential strategy $\sigma$ defined as follows: In the first step the information structure is $I_{\nu_0}$; for every $1<n\le N$, if $x'$ is the realized signal in period $n-1$ then persuasion terminates, and if $x$ is the realized signal then information structure $I_{\nu_{n-1}}$ is used, where $\nu_{n-1}=D(\nu_{n-2},\nu^*)$; if $x$ was realized in all $N$ periods then persuasion terminates.

  Since $D$ satisfies eventual learning, if $N$ is sufficiently large and $x$ was realized in all $N$ periods then $\nu_{N+1}\in B(\nu^*,\epsilon)$. Since $x$ has positive probability in each state $\omega$ and for every $\nu_n$, the lemma follows.
\end{proof}

We can now prove the proposition.

\noindent (If) Fix a 2-action persuasion environment and let $\sigma^*$ be an optimal one-shot persuasion strategy. Recall that $\tau_{\sigma^*}$ and $\tau_{\sigma^*}^D$ are the Bayesian posterior distribution and R's posterior distribution induced by $\sigma^*$. Denote by $T$ the support of $\tau_{\sigma^*}$ and by $T^D$ the support of $\tau_{\sigma^*}^D$. 

First, we claim that there must be $\bar \mu\in T^D$ such that $\hat a(\bar \mu)=a'$. Indeed, since $\mu_0\in Conv(T)$, weak Bayes plausibility implies that $\mu_0\in Conv(T^D)$. But the set $\{\mu\in \Delta(\Omega):~ \hat a(\mu)=a\}$ is convex, so if $\hat a(\mu)=a$ for every $\mu\in T^D$ then we would also have $\hat a(\mu_0)=a$, contradicting the assumption that at the prior the optimal action for R is $a'$. 

Next, by the definition of a 2-action environment there is $\mu^*\in \Delta^o(\Omega)$ such that $\hat a(\mu^*)=a$, which implies that we can choose $\mu^*$ so that R strictly prefers $a$ to $a'$ at $\mu^*$. Let $\epsilon>0$ be small enough so that $\hat a(\mu)=a$ whenever $\mu\in B(\mu^*,\epsilon)$. By Lemma \ref{lemma-converge} there exists a sequential persuasion strategy $\bar \sigma$ such that, starting from the prior $\bar \mu$, induces a posterior in $B(\mu^*,\epsilon)$ with positive probability conditional on every $\omega\in \Omega$.    

Finally, consider the following sequential persuasion strategy. In the first period the optimal one-shot strategy $\sigma^*$ is used. If the posterior of $R$ is not $\bar \mu$ then persuasion terminates; and if the posterior is $\bar \mu$ then $\bar \sigma$ is used. With this strategy the probability that R chooses $a$ is strictly higher than with $\sigma^*$, so S gains from sequential persuasion.

\bigskip

\noindent (Only if) Suppose that $D$ is not weakly Bayes plausible. Then there exist a prior $\mu_0$ and a collection of posteriors $\{\mu_1,\ldots,\mu_K\}$ such that $\mu_0\in Conv(\{\mu_1,\ldots,\mu_K\})$ but $\mu_0\notin Conv(\{D(\mu_0,\mu_1),\ldots,D(\mu_0,\mu_K)\})$. Let $x\in\RR^\Omega$ be a vector that separates $\mu_0$ from $Conv(\{D(\mu_0,\mu_1),\ldots,D(\mu_0,\mu_K)\})$, namely,
$$\sum_\omega x(\omega)\mu_0(\omega)< \sum_\omega x(\omega)D(\mu_0,\mu_k)(\omega), ~~~ k=1,\ldots,K.$$

Consider a 2-action persuasion environment in which the utility function of R is given by $u(\omega,a)=x(\omega)$ and $u(\omega,a')=y$ for every $\omega$, where 
$$y=0.5 \sum_\omega x(\omega)\mu_0(\omega)+ 0.5 \min_{1\le k\le K}\sum_\omega x(\omega)D(\mu_0,\mu_k)(\omega).$$
Then, by construction, $\hat a(\mu_0)=a'$ and $\hat a(D(\mu_0,\mu_k))=a$ for each $k=1,\ldots,K$. In particular, all the requirements in Definition \ref{def-2-actions} are satisfied. 

Let $\sigma$ be a one-shot persuasion strategy such that the support of $\tau_\sigma$ is $\{\mu_1,\ldots,\mu_K\}$. Existence of such a strategy is guaranteed by the fact that $\mu_0\in Conv(\{\mu_1,\ldots,\mu_K\})$. When S uses $\sigma$ R chooses $a$ with probability 1, so S cannot gain from sequential persuasion.  
\hfill \qed


\bigskip

\noindent \textbf{Proof of Proposition \ref{prop-transparent-linear}} 

\smallskip

We start with the following lemma.

\begin{lemma}\label{lemma-full-dim}
    Fix a linear persuasion environment and suppose that $D$ is interior. For every $\mu\in\Delta^o(\Omega)$ there exists a state $\underline\omega\in \Omega$ such that $\beta\cdot\underline{\omega} < \beta\cdot \EE_{\omega\sim D(\mu,\delta_{\underline{\omega}})} [\omega]$.
\end{lemma}

\begin{proof}
Fix $\mu$ and let $\underline \omega$ be a state that satisfies $\beta\cdot \underline\omega\le \beta \cdot \omega$ for every $\omega\in \Omega$. Since $Conv(\Omega)$ is a set of full dimension there must be some $\bar\omega$ such that $\beta\cdot \underline\omega< \beta \cdot \bar \omega$ (otherwise all the states are on the same hyperplane orthogonal to $\beta$). We thus have
\begin{eqnarray*}
\beta\cdot \EE_{\omega\sim D(\mu,\delta_{\underline\omega})} [\omega] = \sum_{\omega\neq \bar\omega} D(\mu,\delta_{\underline\omega})(\omega)(\beta\cdot \omega)+ D(\mu,\delta_{\underline\omega})(\bar \omega)(\beta\cdot \bar\omega) > \\
(\beta \cdot \underline\omega)  \sum_{\omega\neq \bar\omega} D(\mu,\delta_{\underline\omega})(\omega) + (\beta\cdot \underline\omega) D(\mu,\delta_{\underline\omega})(\bar \omega) = \beta\cdot\underline \omega,
\end{eqnarray*}
where the strict inequality is due to the assumption that $D$ is interior and hence that $D(\mu,\delta_{\underline\omega})(\bar \omega)>0$.
\end{proof}

To prove the proposition we consider two possible cases.

\medskip

\noindent \underline{Case i}: $V_D(\mu_0)=\beta\cdot \EE_{\omega\sim\mu_0} [\omega]$ (S cannot increase their payoff with one-shot persuasion).

Denote 
$$C^-=\{\omega :~ \beta\cdot\omega < \beta\cdot \EE_{\omega'\sim D(\mu_0,\delta_{\omega})} [\omega']\}$$
and
$$C^+=\{\omega :~ \beta\cdot\omega \ge \beta\cdot \EE_{\omega'\sim D(\mu_0,\delta_{\omega})} [\omega']\}.$$
Notice that $C^-\neq \emptyset$ by Lemma \ref{lemma-full-dim} applied to $\mu=\mu_0$. 

Consider the following sequential persuasion strategy. In the first period a fully revealing signal is sent; if $\omega\in C^-$ then the strategy terminates; and if $\omega\in C^+$ then $n$ additional fully revealing signals are sent after which the strategy terminates. 

We claim that if $n$ is sufficiently large then this strategy increases S's payoff above the no information payoff $\beta\cdot \EE_{\omega\sim\mu_0} [\omega]$. Indeed, since $D$ satisfies eventual learning of states, we have that $\lim_n D^n(\mu_0,\delta_{\omega})=\delta_{\omega}$ for each $\omega\in C^+$. Therefore, as $n\to \infty$ the limit of the expected payoff to S with the above strategy is given by 
\begin{eqnarray*}
\sum_{\omega\in C^-} \mu_0(\omega) (\beta\cdot \EE_{\omega'\sim D(\mu_0,\delta_{\omega})} [\omega']) + \sum_{\omega\in C^+} \mu_0(\omega) (\beta\cdot \omega) &>& \\
\sum_{\omega\in C^-} \mu_0(\omega) (\beta\cdot\omega) &+& \sum_{\omega\in C^+} \mu_0(\omega) (\beta\cdot \omega) = \beta\cdot \EE_{\omega\sim\mu_0} [\omega],
\end{eqnarray*}
where the inequality is strict by the definition of the set $C^-$ and its non-emptiness. 

\medskip

\noindent \underline{Case ii}: $V_D(\mu_0)>\beta\cdot \EE_{\omega\sim\mu_0} [\omega]$ (S can increase their payoff with one-shot persuasion). 

Let $\sigma^*$ be an optimal one-shot persuasion strategy. We claim that there is $\bar \mu\in T$ ($T$ is the support of $\tau_{\sigma^*}$) such that 
\begin{equation}\label{eqn-bar-mu}
    \beta\cdot (\EE_{\omega\sim D(\mu_0,\bar \mu)}[\omega]-\EE_{\omega\sim \bar \mu}[\omega])\le 0.
\end{equation}
Indeed, since $D$ is balanced there are non-negative weights $\{\gamma_\mu\}_{\mu\in T}$ summing up to 1 such that 
$$\sum_{\mu\in T} \gamma_\mu D(\mu_0,\mu)=\sum_{\mu\in T} \gamma_\mu \mu.$$
This implies that 
$$\EE_{\omega\sim \sum_{\mu\in T} \gamma_\mu D(\mu_0,\mu))}[\omega] = \EE_{\omega\sim \sum_{\mu\in T} \gamma_\mu \mu}[\omega],$$
and therefore that 
$$\sum_{\mu\in T} \gamma_\mu\EE_{\omega\sim  D(\mu_0,\mu)}[\omega] = \sum_{\mu\in T} \gamma_\mu\EE_{\omega\sim  \mu}[\omega].$$
In other words, the set $Conv(\{\EE_{\omega\sim D(\mu_0, \mu)}[\omega]-\EE_{\omega\sim \mu}[\omega] :~ \mu\in T\})$ contains the origin of $\RR^K$. But then we cannot have $\beta\cdot (\EE_{\omega\sim D(\mu_0, \mu)}[\omega]-\EE_{\omega\sim \mu}[\omega])>0$ for every $\mu\in T$, proving that $\bar\mu$ as claimed exists.

Now, similarly to the above case i, define the sets 
$$C^-=\{\omega :~ \beta\cdot\omega <\beta\cdot \EE_{\omega'\sim D(D(\mu_0,\bar\mu),\delta_{\omega})} [\omega']\}$$
and
$$C^+=\{\omega :~ \beta\cdot\omega \ge \beta\cdot \EE_{\omega'\sim D(D(\mu_0,\bar\mu),\delta_{\omega})} [\omega']\}.$$
Applying Lemma \ref{lemma-full-dim} with $\mu=D(\mu_0,\bar\mu)$ we get that $C^-\neq \emptyset$.

Consider the following sequential persuasion strategy: In the first period $\sigma^*$ is used. If the realized posterior is any $\mu\neq \bar \mu$ then persuasion terminates. If $\bar \mu$ is the realized posterior then persuasion continues in a similar way to case (i): First, a fully revealing signal is sent; if $\omega\in C^-$ then the strategy terminates; and if $\omega\in C^+$ then $n$ additional fully revealing signals are sent after which the strategy terminates. 

Using eventual learning of states, as $n\to \infty$ the limit of the expected payoff to S when using this strategy is given by 
\begin{eqnarray*}
& & \sum_{\mu\in T\setminus \{\bar \mu\}}\tau_{\sigma^*}(\mu)(\beta\cdot \EE_{\omega\sim D(\mu_0,\mu))}[\omega]) +\\ 
& & \tau_{\sigma^*}(\bar \mu) \left[ \sum_{\omega\in C^-} \bar\mu(\omega) (\beta\cdot \EE_{\omega'\sim D(D(\mu_0,\bar\mu),\delta_{\omega})} [\omega']) + \sum_{\omega\in C^+} \bar\mu(\omega) (\beta\cdot \omega)\right] >\\
& & \sum_{\mu\in T\setminus \{\bar \mu\}}\tau_{\sigma^*}(\mu)(\beta\cdot \EE_{\omega\sim D(\mu_0,\mu))}[\omega]) + \tau_{\sigma^*}(\bar \mu) \left[ \sum_{\omega\in C^-} \bar\mu(\omega) (\beta\cdot \omega) + \sum_{\omega\in C^+} \bar\mu(\omega) (\beta\cdot \omega)\right] =\\
& & \sum_{\mu\in T\setminus \{\bar \mu\}}\tau_{\sigma^*}(\mu)(\beta\cdot \EE_{\omega\sim D(\mu_0,\mu))}[\omega]) + \tau_{\sigma^*}(\bar \mu)(\beta\cdot \EE_{\omega\sim \bar \mu}[\omega]) \ge\\
& & \sum_{\mu\in T\setminus \{\bar \mu\}}\tau_{\sigma^*}(\mu)(\beta\cdot \EE_{\omega\sim D(\mu_0,\mu))}[\omega]) + \tau_{\sigma^*}(\bar \mu)(\beta\cdot \EE_{\omega\sim D(\mu_0,\bar\mu)}[\omega]) = V_D(\mu_0),
\end{eqnarray*}
where the strict inequality is by the definition of the set $C^-$ and its non-emptiness, the weak inequality is by (\ref{eqn-bar-mu}), and the last equality is by the optimality of $\sigma^*$ for the one-shot persuasion problem. It follows that for sufficiently large $n$ the expected payoff for S under this strategy exceeds $V_D(\mu_0)$.
\hfill \qed


\bigskip

\noindent \textbf{Proof of Proposition \ref{prop-common-pref}} 

\smallskip

(If) Consider the sequential strategy in which S sends fully revealing signals for $N$ periods and then stops. If the true state is $\omega$, then the Bayesian posterior after each of the signals will be $\delta_\omega$. Thus, by eventual learning of states, for any $\epsilon>0$ we can choose $N$ sufficiently large so that R's belief is within $\epsilon$ from $\delta_\omega$. Since S and R have the same preferences $\hat v$ is the pointwise maximum of a family of linear functions and hence continuous. Therefore, as $N\to \infty$ the payoff for S converges to $\sum_\omega \mu_0(\omega)\hat v(\delta_\omega)$, which is precisely the right-hand side of (\ref{eqn-common-pref}).

(Only if) Consider any sequential strategy $\sigma$. The payoff for S from choosing $\sigma$ is
\begin{eqnarray*}
 \int \EE_{\omega\sim\mu}[v(\omega, \hat a(\mu'))] d\lambda_\sigma^D(\mu,\mu') &=&  \int \EE_{\omega\sim\mu}[u(\omega, \hat a(\mu'))] d\lambda_\sigma^D (\mu,\mu')\le \\
 & & \int \EE_{\omega\sim\mu}[u(\omega, \hat a(\mu))] d\tau_\sigma(\mu) =  \int \hat{v}(\mu) d\tau_\sigma(\mu),
\end{eqnarray*}
where the first equality is from $v=u$, the inequality is by the definition of $\hat a$, and the last equality follows from $u=v$ and the definition of $\hat v$. Note that the last expression is the payoff for S given strategy $\sigma$ when R is Bayesian. Since S and R have the same preferences this is bounded above by the full-information payoff on the right-hand side of (\ref{eqn-common-pref}). Thus, if (\ref{eqn-common-pref}) does not hold then there is no way for S to gain from sequential persuasion. 
\hfill \qed


\bigskip

\noindent \textbf{Proof of Corollary \ref{coro-common-pref}} 

\smallskip

By Proposition \ref{prop-common-pref} it is enough to show that if (\ref{eqn-empty-intersect}) holds then S cannot achieve the Bayesian full-information payoff with one period of persuasion. Getting the full information payoff requires that in each state $\omega$ only actions that are optimal at $\omega$ are chosen with positive probability. However, condition (\ref{eqn-empty-intersect}) implies that any action optimal at $\bar\omega$ is not optimal at any belief in the set $\{D(\mu_0,\mu) :~ \mu\in\Delta(\Omega)\}$. Since with any information structure the beliefs of R are always inside this set, it follows that non-optimal actions are chosen with positive probability in state $\bar \omega$. 

\hfill \qed


\bigskip

\noindent \textbf{Proof of Proposition \ref{prop-CS}} 

\smallskip

The proof is similar to that of Proposition \ref{prop-transparent-linear}, although here we do not need to consider the two cases separately. We assume that the bias is positive, $b>0$; the case $b<0$ follows by a symmetric argument, and the case $b=0$ follows from Corollary \ref{coro-common-pref}. The following lemma will be useful below.

\begin{lemma}\label{lemma-CS}
Fix a Crawford-Sobel environment and let $\sigma$ be some information structure. Suppose that $D$ is interior and balanced. Then there exists $\bar\mu\in T:=Supp(\tau_\sigma)$ that satisfies one of the following conditions:\\
(i) $\EE_{\omega\sim\bar\mu}[\omega]> \EE_{\omega\sim D(\mu_0,\bar\mu)}[\omega]$; or\\
(ii) $\EE_{\omega\sim\bar\mu}[\omega]=\EE_{\omega\sim D(\mu_0,\bar\mu)}[\omega]$ and $\bar\mu$ is not deterministic.
\end{lemma}

\begin{proof}
Since $D$ is balanced there are non-negative weights $\{\gamma_\mu\}_{\mu\in T}$ summing up to 1 such that 
$$\sum_{\mu\in T} \gamma_\mu D(\mu_0,\mu)=\sum_{\mu\in T} \gamma_\mu \mu.$$
Let $T'=\{\mu\in T ~:~ \gamma_\mu>0\}\subseteq T$, so that 
\begin{equation}\label{eqn-CS-T'}
\sum_{\mu\in T'} \gamma_\mu D(\mu_0,\mu)=\sum_{\mu\in T'} \gamma_\mu \mu
\end{equation}
still holds and all the coefficients are strictly positive. As in the proof of Proposition \ref{prop-transparent-linear} above, this implies that 
$$\sum_{\mu\in T'} \gamma_\mu\EE_{\omega\sim  D(\mu_0,\mu)}[\omega] = \sum_{\mu\in T'} \gamma_\mu\EE_{\omega\sim  \mu}[\omega].$$

Now, if there is $\bar\mu\in T'$ such that $\EE_{\omega\sim\bar\mu}[\omega]> \EE_{\omega\sim D(\mu_0,\bar\mu)}[\omega]$ then we are done. Otherwise, it must be that $\EE_{\omega\sim\mu}[\omega]=\EE_{\omega\sim D(\mu_0,\mu)}[\omega]$ for all $\mu\in T'$. We claim that in this case it is impossible that all the beliefs in $T'$ are deterministic. Indeed, assume by contradiction that this is the case. The vector on the left-hand side of (\ref{eqn-CS-T'}) has strictly positive coordinates since $D$ is interior. Thus, for (\ref{eqn-CS-T'}) to hold it must be that $T'=\{\delta_\omega\}_{\omega\in\Omega}$, i.e., $T'$ must contain the Dirac measures of all the states. But for the largest state, say $\bar\omega$, it is impossible that $\delta_{\bar\omega}\in T'$ since then $\EE_{\omega\sim\delta_{\bar\omega}}[\omega]= \bar\omega > \EE_{\omega\sim D(\mu_0,\bar\mu)}[\omega]$, where the inequality follows from $D$ being interior. 
\end{proof}

We now move on to the proof of the proposition. Let $\sigma^*$ be an optimal one-shot persuasion strategy (possibly degenerate), and let $T$ be the support of $\tau_{\sigma^*}$. Choose $\bar\mu\in T$ that satisfies either (i) or (ii) of Lemma \ref{lemma-CS}. Consider the following sequential persuasion strategy: In the first period $\sigma^*$ is used. If the realized posterior is any $\mu\neq \bar \mu$ then persuasion terminates. If $\bar \mu$ is the realized posterior then $n$ fully revealing signals are sent, after which the strategy terminates. Since $D$ satisfies eventual learning of states, as $n$ grows to infinity the expected payoff of S under this strategy converges to
\begin{equation}\label{eqn-CS1}
-\sum_{\mu\in T\setminus \{\bar \mu\}}\tau_{\sigma^*}(\mu)\EE_{\omega\sim \mu}\left[(\EE_{\omega'\sim D(\mu_0,\mu)}[\omega']-\omega-b)^2\right] - \tau_{\sigma^*}(\bar\mu)b^2. 
\end{equation}
If on the other hand S stops immediately after $\sigma^*$ then their payoff is 
\begin{equation}\label{eqn-CS2}
-\sum_{\mu\in T}\tau_{\sigma^*}(\mu) \EE_{\omega\sim \mu}\left[(\EE_{\omega'\sim D(\mu_0,\mu)}[\omega']-\omega-b)^2\right]. 
\end{equation}

To prove the proposition we need to argue that (\ref{eqn-CS1}) is strictly larger than (\ref{eqn-CS2}), which is equivalent to showing that
$$b^2< \EE_{\omega\sim \bar\mu}\left[(\EE_{\omega'\sim D(\mu_0,\bar\mu)}[\omega']-\omega-b)^2\right].$$
After some simple algebra the latter inequality becomes 
$$2b\left(\EE_{\omega'\sim D(\mu_0,\bar\mu)}[\omega']-\EE_{\omega\sim \bar\mu}\left[\omega\right] \right)< \EE_{\omega\sim \bar\mu}\left[(\EE_{\omega'\sim D(\mu_0,\bar\mu)}[\omega']-\omega)^2\right].$$
Now, if $\bar\mu$ satisfies property (i) of Lemma \ref{lemma-CS} then the left-hand side is strictly negative (recall that $b>0$) and the right-hand side is non-negative; and if $\bar\mu$ satisfies property (ii) of the lemma then the left-hand side is zero and the right-hand side is strictly positive; in either case, strict inequality holds.


\bibliographystyle{plainnat}
\bibliography{Sequential_BP}


\end{document}